\documentclass[10pt,twoside,english]{article}
\usepackage[T1]{fontenc}
\usepackage[utf8]{inputenc}
\usepackage{amsmath,amsthm,amssymb}
\usepackage[dvips]{graphicx}
\usepackage{array}
\usepackage[mathscr]{eucal}
\usepackage{eurosym}
\usepackage{subfigure}
\usepackage{color}

\usepackage{numprint}
\npstyleenglish

\usepackage{booktabs}

\usepackage[table]{xcolor}

\usepackage{tikz}
\usetikzlibrary{shapes,arrows}

\definecolor{shadecolor}{HTML}{cbe1ee}

\newtheorem{Theorem}{Theorem}[section]

\newtheorem{Proposition}[Theorem]{Proposition}

\newtheorem{Definition}[Theorem]{Definition}

\theoremstyle{remark}

\numberwithin{equation}{section}
\numberwithin{figure}{section}
\numberwithin{table}{section}

\newcommand{\Esp}[1]{\mathrm{E}\! \left[ #1 \right]}
\newcommand{\Var}[1]{\mathrm{Var}\! \left[ #1 \right]}
\newcommand{\Cov}[1]{\mathrm{Cov}\! \left[ #1 \right]}
\newcommand{\Prr}[1]{\mathrm{Pr}\! \left( #1 \right)}

\begin{document}

\title{A Claim Score for Dynamic Claim Counts Modeling}
\author{\textsc{Jean-Philippe Boucher$^{\ddag}$ and Mathieu Pigeon$^{\star}$}}

\maketitle
\vspace{-1cm}

\begin{center}
Quantact / Département de mathématiques, UQAM.
Montréal, Québec, Canada. \\
 \textit{Emails: }\texttt{\textit{\small $^{\ddag}$boucher.jean-philippe@uqam.ca}}{\small {}}\\ \texttt{\textit{\small $^{\star}$pigeon.mathieu.2@uqam.ca}}{\small {}} \\
\par
\end{center}

\begin{abstract}
We develop a claim score based on the Bonus-Malus approach proposed by \cite{Boucher2014}.  We compare the fit and predictive ability of this new model with various models for of panel count data. In particular, we study in more details a new dynamic model based on the Harvey-Fernandès (HF) approach, which gives different weight to the claims according to their date of occurrence. We show that the HF model has serious shortcomings that limit its use in practice. In contrast, the Bonus-Malus model does not have these defects. Instead, it has several interesting properties: interpretability, computational  advantages and ease of use in practice. We believe that the flexibility of this new model means that it could be used in many other actuarial contexts. Based on a real database, we show that the proposed model generates the best fit and one of the best predictive capabilities among the other models tested.

\end{abstract}

\section{Introduction and Motivation}
 \label{sec:introduction}

In a recent paper, \cite{Boucher2014} proposed a new approach to merit pricing for the number of claims in automobile insurance. Instead of assuming a random effect or a copula to model the dependence between all the contracts of a policyholder, the authors directly included a bonus-malus system (BMS) in the modeling. Consequently, all the past claims experience is summarized into a single numerical value: the current level of the BMS or what we call a \textit{claim score} (similar to a credit score). Besides strengthening links between the theory and what is currently applied in practice for pricing, the proposed approach showed an interesting flexibility that can be used, for example, to include some legal constraints in the ratemaking. However, despite the advantages of the approach, the parameter estimation procedure proposed by \cite{Boucher2014} is quite cumbersome and time-consuming. 
 
In this paper, we propose a straightforward but significant modification to this Bonus-Malus Scales Model for Panel Count Data (BMS-panel model). It changes the way the parameters are estimated, drastically reducing both the complexity and the time of the calibration procedure. Taking advantage of this new flexibility of the model, we perform an exhaustive analysis of the model, where we do not set restriction on the size of the parameters space. In addition to showing some of the properties of the model, we compare the BMS-panel model with various panel data count models, including a new model that we introduce. We show that the fit of the BMS-panel model is the best among the considered distributions, and we analyze its prediction power.

Because the BMS-panel model has a Markovian property that greatly facilitates its use, we believe that the bonus-malus approach could be an interesting option for modeling the complex structure of claim experience in actuarial sciences. For example, in situations where we want to model dependence between cars from the same contract, or when we want to link the frequency and the severity of claims, the \textit{claims score} produced by the BMS model could be considered.  This new paradigm could then replace complex approaches such as using a series of correlated random effects (see \cite{AA2016}) or a multiple hierarchical copula (see \cite{Shi2016}).

The paper is constructed according to the following structure. In Section~\ref{sec:ratemaking}, we briefly present the main ratemaking approaches, where we introduce a new dynamic count 
distribution for panel data.  In Section~\ref{ssec:BMS}, we review the BMS-panel model, we introduce a simplified version and we highlight some properties of the model.
In Section~\ref{sec:illustrations}, based on a database from a property and casualty insurance company, we calibrate and compare the proposed models.  A very wide range of structural parameters of the 
BMS-panel model is tested. Finally, we conclude and present some promising generalizations in Section~\ref{sec:conclusion}.
 
\section{Ratemaking Techniques}
\label{sec:ratemaking}
 
 We consider a property and casualty portfolio of $M$~policyholders observed over several years. For each contract~$i$, $i = 1, \ldots, M$, we define $N_{i,t}$, a discrete random variable counting the number of claims for the policy period $t$ and $\boldsymbol{X}_{i,t}$ a column-vector containing available explanatory factors at the beginning of period $t$. In this vector, we may include $d_{i,t}$, a scalar measuring the risk exposure. Through this project, we assume that the primary purpose of a ratemaking (or pricing) model is to provide a prediction for
 \begin{align*}
   &\Esp{N_{i,T_i+1}|\underbrace{N_{i,1}, \ldots, N_{i,T_i}}_{\boldsymbol{N}_{i,T_i}}, \boldsymbol{X}_{i,1}, \ldots, \boldsymbol{X}_{i,T_i + 1}},
 \end{align*}
where all policyholders are independent and $T_i$ is the latest observed period for policyholder~$i$. Our main objective is to investigate the strengths and weaknesses of the following classes of models:
 \begin{itemize}
 \item cross-section data models, for which independence is assumed between annual contracts for a policyholder (see subsection~\ref{ssec:cross});
 \item panel data models, for which we suppose dependence between all contracts written with a policyholder (see subsection~\ref{ssec:panel}); and 
 \item BMS-panel models, for which at least some of the past information is summarized using a bonus-malus system (BMS) (see section~\ref{ssec:BMS}).
 \end{itemize}
 
\subsection{Cross-section data models}
\label{ssec:cross}

For cross-sectional models, we have independence between all policyholders as well as between all contracts so we can write
\begin{align*}
 \Prr{N_{i,t+1} = n| \boldsymbol{N}_{i,t}, \boldsymbol{X}_{i,1}, \ldots, \boldsymbol{X}_{i,t+1}} &= \Prr{N_{i,t+1} = n| \boldsymbol{X}_{i,t+1}}
\end{align*}
and we can simplify our prediction problem in the following way:
\begin{align*}
  \Esp{N_{i,t+1}|\boldsymbol{N}_{i,t}, \boldsymbol{X}_{i,1}, \ldots, \boldsymbol{X}_{i,t + 1}} &= \Esp{N_{i, t + 1}| \boldsymbol{X}_{i,t + 1}} = \lambda\left(\boldsymbol{X}_{i,t+1}\right),
\end{align*}
where $\lambda()$ is a function. Traditionally, in risk classification, we assume a log linear relationship between the mean parameter and the policyholder's and/or claim's characteristics such as sex, age, marital status, etc. (see \cite{Denuit}).

The base model is usually the Poisson distribution, which is part of the exponential linear family and has useful and well-known statistical properties (see \cite{McCullaghNelder1989} or \cite{frees2014}). The probability mass function is $\Prr{N_{i,t} = n| \boldsymbol{X}_{i,t}} = \left(\lambda_{i,t}\right)^n\exp\left(-\lambda_{i,t}\right)/n!$, $n = 0, 1, 2, \ldots$ and $0$ elsewhere, where $\lambda_{i,t} = d_{i,t} \exp\left(\boldsymbol{X}_{i,t}'\boldsymbol{\beta}\right)$ and $\boldsymbol{\beta}$ is a column vector containing model parameters. Finally, we note
\begin{align*}
  \pi_{i,t+1}^{\text{Poi}} &= \Esp{N_{i, t+1}| \boldsymbol{X}_{i,t+1}} =  \lambda_{i,t+1}.
\end{align*}
This equation refers to an annual premium, when the cost of each claim is $1$. Because this premium does not depend on the past claim experience, we usually call $\pi_{i,t+1}$ an \emph{a priori premium}.

This Poisson distribution implies equidispersion, i.e., $\Esp{N_{i, t}| \boldsymbol{X}_{i,t}} = \Var{N_{i, t}| \boldsymbol{X}_{i,t}}$ which is, usually, a too strong assumption in non-life ratemaking. In an effort to overcome this issue, we consider the Negative Binomial (NB) distribution, which is one of the most commonly used alternatives to the Poisson model. To facilitate data analysis, we limit ourselves to the simplest forms of the NB distribution but someone can consult \cite{Wink}.  The most intuitive way to construct an NB distribution from a Poisson distribution is to introduce a random heterogeneity term in the mean parameter. Let $\Theta$ be a random variable following a Gamma$(\alpha = 1/\tau, \gamma = 1/\tau)$ distribution ($\Esp{\Theta} = 1$) with probability density function $f_\Theta(\theta) = \tau^\tau \theta^{\tau - 1}\exp\left(-\theta\tau\right)/\Gamma(\tau)$, $\theta > 0$ and $0$ elsewhere and assume that $\left(N_{i,t}|\Theta, \boldsymbol{X}_{i,t}\right) \sim \text{Poisson}\left(\lambda_{i,t}\theta\right)$. Therefore, the random variable $\left(N_{i,t}|\boldsymbol{X}_{i,t}\right)$ follows a Negative Binomial distribution of type $2$ (NB2) with parameters $\tau$ and $\lambda_{i,t}$ and probability mass function given by
\begin{align*}
  \Prr{N_{i,t} = n| \boldsymbol{X}_{i,t}} &= \frac{\Gamma(n + 1/\tau)}{\Gamma(n+1)\Gamma(1/\tau)}\left(\frac{\lambda_{i,t}}{1/\tau + \lambda_{i,t}}\right)^{n}\left(\frac{1/\tau}{1/\tau + \lambda_{i,t}}\right)^{1/\tau}, \qquad n = 0, 1, 2, \ldots
\end{align*}
and $0$ elsewhere. We can directly obtain $\Esp{N_{i,t}} = \lambda_{i,t}$ and $\Var{N_{i,t}} = \lambda_{i,t} + \tau^2 \lambda_{i,t}$.

We also consider a slightly different version of the Negative Binomial distribution (NB1) with parameters $\lambda_{i,t}$ and $\tau$ for which the probability mass function is
\begin{align*}
    \Prr{N_{i,t} = n| \boldsymbol{X}_{i,t}} &= \frac{\Gamma\left(n + \frac{\lambda_{i,t}}{\tau}\right)}{\Gamma(n+1)\Gamma\left(\frac{\lambda_{i,t}}{\tau}\right)}\left(1 + \tau\right)^{-\lambda_{i,t}/\tau}\left(1 + \frac{1}{\tau}\right)^{-n}, \qquad n = 0, 1, 2, \ldots
\end{align*}
and $0$ elsewhere. We observe that $\Var{N_{i,t}} = \lambda_{i,t} + \tau\lambda_{i,t} = \phi\lambda_{i,t}$, which corresponds to the variance function of the overdispersed Poisson in the generalized linear model (GLM) framework. In all cases, parameters can be easily estimated using a maximum-likelihood procedure.

The premium predicted by both the NB1 and NB2 model is
\begin{align*}
  \pi_{i,t+1}^{\text{NB}} &= \Esp{N_{i, t+1}| \boldsymbol{X}_{i,t+1}} = \lambda_{i,t+1},
\end{align*}
where independence between contracts of the same insured is still assumed.

 \subsection{Panel data models}
 \label{ssec:panel}

 Panel data models assume some dependence between all annual contracts belonging to a single policyholder. We need a model for the (conditional) random vector $\begin{bmatrix}\boldsymbol{N}_{i,t} | \boldsymbol{X}_{i,1}, \ldots, \boldsymbol{X}_{i,t}\end{bmatrix}$, $t = 1, 2, \ldots$ in order to predict $\Esp{N_{i,T_i+1}|\boldsymbol{N}_{i,T_i}, \boldsymbol{X}_{i,1}, \ldots, \boldsymbol{X}_{i,T_i + 1}}$. There are plenty of models for the time dependence between random variables, e.g., conditional models, marginal models, and subject-specific models, but it has been shown that random effects models were the best suited for non-life insurance data (see \cite{Boucher2009}). In a ratemaking model, an individual random effect may capture variability caused by the lack of information on some important classification variables such as road rage and drug use. Let $\Theta$ denote this random effect. Conditionally on $\Theta$, all contracts of the same insured are supposed independent. The joint probability mass function is defined by
 \begin{align*}
   &\Prr{N_{i,1} = n_{i, 1}, \ldots, N_{i,t} =  n_{i,t} | \boldsymbol{X}_{i,1}, \ldots, \boldsymbol{X}_{i,t}}\\
   &= \int_{-\infty}^\infty \Prr{N_{i,1} = n_{i,1}, \ldots, N_{i,t} =  n_{i,t}|\theta_i, \boldsymbol{X}_{i,1}, \ldots, \boldsymbol{X}_{i,t}}\,dG(\theta_i|\boldsymbol{X}_{i,1}, \ldots, \boldsymbol{X}_{i,t}), \\
&=\int_{-\infty}^\infty \left(\prod_{k = 1}^t \Prr{N_{i,k} = n_{i,k}|\theta_i, \boldsymbol{X}_{i,1}, \ldots, \boldsymbol{X}_{i,t}}\right)\,dG(\theta_i).
 \end{align*}
where $G(\theta_i)$ is the cumulative distribution function of the random effect.  We assume that the distribution of the random effect $\theta_i$ does not depend on covariates $\boldsymbol{X}$; see \cite{Boucher2006} for a discussion about this conventional assumption in actuarial science. Finally, note that the joint distribution can also be expressed as the product of all predictive distributions of each insurance contract for insured $i$:
    \begin{align*}
      &\Prr{N_{i,1} = n_{i,1}, \ldots, N_{i,t} =  n_{i,t}|\boldsymbol{X}_{i,1}, \ldots, \boldsymbol{X}_{i,t}}\\
      &=\Prr{N_{i,1} = n_{i,1}|\boldsymbol{X}_{i,1}}  \Prr{N_{i,2} = n_{i,2}|N_{i,2} = n_{i,1}, \boldsymbol{X}_{i,1}, \boldsymbol{X}_{i,2}}  \\
      &\phantom{=}  \times \cdots \times \Prr{N_{i,t} = n_{i,t}| N_{i,1} = n_{i,1}, \ldots, N_{i,t-1} =  n_{i,t-1},\boldsymbol{X}_{i,1}, \ldots, \boldsymbol{X}_{i,t}}.
    \end{align*}
 
\subsubsection{Negative Multinomial}  
 
Following the construction of the NB2, the simplest random effects model is given by
 \begin{align*}
   \left(N_{i,t}|\Theta_i = \theta_i, \boldsymbol{X}_{i,1}, \ldots, \boldsymbol{X}_{i,t}\right) &\sim \text{Poisson}\left(\theta_i\lambda_{i,t}\right)
 \end{align*}
 and $\Theta_i \sim \text{Gamma}(\alpha = \kappa, \gamma = \kappa)$ which lead to

\begin{align*}
  \Prr{N_{i,1} = n| \boldsymbol{X}_{i,1}} &= \frac{\Gamma(n + \kappa)}{\Gamma(n+1)\Gamma(\kappa)}\left(\frac{\lambda_{i,1}}{\kappa + \lambda_{i,1}}\right)^{n}\left(\frac{\kappa}{\kappa + \lambda_{i,1}}\right)^\kappa \\
  \Prr{N_{i,t+1} = n| \boldsymbol{N}_{i,t}, \boldsymbol{X}_{i,1}, \ldots, \boldsymbol{X}_{i,t+1}} &= \frac{\Gamma(n + \alpha)}{\Gamma(n+1)\Gamma(\alpha)}\left(\frac{\lambda_{i,t+1}}{\gamma + \lambda_{i,t+1}}\right)^{n}\left(\frac{\gamma}{\gamma + \lambda_{i,t+1}}\right)^\alpha, 
\end{align*}
 with $\alpha = \kappa + \sum_{k = 1}^{t} n_{i, k}$, $\gamma = \kappa + \sum_{k = 1}^{t}\lambda_{i, k}$, and 
 $\lambda_{i,k} = d_{i,k} \exp\left(\boldsymbol{X}_{i,k}'\boldsymbol{\beta}\right)$.  Both probability distributions have a NB2 form (for all $t$), and then can be simply noted as $\text{NB2}_{i,t}(\lambda_{i,t}, \alpha, \gamma)$. 

The predictive distribution comes from the classical credibility theory (see \cite{BG}), where the \emph{a posteriori} distribution of the heterogeneity term $\Theta_i$ is a Gamma distribution with updated parameters $\kappa + \sum_{k = 1}^t n_{i, k}$ and $\kappa + \sum_{k = 1}^t \lambda_{i,k}$. The joint distribution, called Multivariate Negative Binomial distribution (MVNB), or Negative Multinomial, is often applied in non-life insurance. Again, parameters can be estimated using a maximum likelihood method. Therefore, the requested prediction is
    \begin{align*}
         \pi_{i,1}^{\text{MVNB}} &= \Esp{N_{i,1}|\boldsymbol{X}_{i,1}} = \lambda_{i, 1}\left(\frac{\kappa}{\kappa}\right) = \lambda_{i, 1} \\
      \pi_{i, t+1}^{\text{MVNB}} &= \Esp{N_{i,t+1}|\boldsymbol{N}_{i,t}, \boldsymbol{X}_{i,1}, \ldots, \boldsymbol{X}_{i,t + 1}} = \lambda_{i, t + 1}\left(\frac{\kappa + \sum_{k = 1}^{t} n_{i, k}}{\kappa + \sum_{k = 1}^{t}\lambda_{i, k}}\right).
    \end{align*}

Given that each insurance contract of the same policyholder has the same random effects, $N_{i,j}$ and $N_{i,t+j}$ are dependent:
\begin{align*}
\Cov{N_{i,t}, N_{i,t+j}|\boldsymbol{X}_{i,1}, \ldots, \boldsymbol{X}_{i,t+j}} &= \lambda_{i, t} \lambda_{i, t+j} (1/\kappa), \qquad j > 0.
 \end{align*}

\subsubsection{Negative Binomial with random effects}  

As pointed out by \cite{Boucher2008}, the Negative Binomial with random effects Beta is well suited to model the number of claims, resulting in what we call a NBBeta distribution. From the probability mass function of the NB2 distribution, we assume that $(1/\tau)/(\lambda + 1/\tau) \sim \text{Beta}(a, b)$. Then, we have 
\begin{align*}
  \Prr{N_{i,1} = n| \boldsymbol{X}_{i,1}} &= \frac{\Gamma(a+b) \Gamma(a+\lambda_{i,1})\Gamma(b+ n)}
{\Gamma(a) \Gamma(b) \Gamma(a+b+ \lambda_{i,1} + n)}
\frac{\Gamma(\lambda_{i,1} + n)}{\Gamma(\lambda_{i,1}) 
\Gamma(n + 1)} \\
  \Prr{N_{i,t+1} = n| \boldsymbol{N}_{i,t}, \boldsymbol{X}_{i,1}, \ldots, \boldsymbol{X}_{i,t+1}} &= \frac{\Gamma(\alpha+\gamma) \Gamma(\alpha+\lambda_{i,t+1})\Gamma(\gamma+ n)}
{\Gamma(\alpha) \Gamma(\gamma) \Gamma(\alpha+\gamma+ \lambda_{i,t+1} + n)}
\frac{\Gamma(\lambda_{i,t+1} + n)}{\Gamma(\lambda_{i,t+1}) 
\Gamma(n + 1)} \\
\end{align*}
 with $\alpha = a + \sum_{k = 1}^{t} \lambda_{i, k}$, $\gamma = b + \sum_{k = 1}^{t}n_{i, k}$ and $\lambda_{i,k} = d_{i,k} \exp\left(\boldsymbol{X}_{i,k}'\boldsymbol{\beta}\right)$.  Both probability distributions have an NBB form (for all $t$), and then can be simply noted as
$\text{NBB}_{i,t}(\lambda_{i,t}, \alpha, \gamma)$. \\
 
 We can easily show that  $\Esp{N_{i,t}} = \lambda_{i,t}\left(b/(a - 1)\right)$, and
 \begin{align*}
   \Var{N_{i,t}} &= \lambda_{i,t} \frac{(a+b-1)b}{(a-1)(a-2)} + \lambda_{i,t}^2 \left(\frac{(b+1)b}{(a-1)(a-2)} - \frac{b^2}{(a-1)^2} \right)
 \end{align*}
 and derive the \textit{a posteriori} distribution of the random effets, which is a Beta distribution with updated parameters $\sum_t \lambda_{i,t} + a$ and $\sum_t n_{i,t} + b$. Thus, we have
	\begin{align*}
         \pi_{i,1}^{\text{NBB}} &= \Esp{N_{i,1}|\boldsymbol{X}_{i,1}} = \lambda_{i, 1}\left(\frac{b}{a-1}\right) \\
      \pi_{i, t+1}^{\text{NBB}} &= \Esp{N_{i,t+1}|\boldsymbol{N}_{i,t}, \boldsymbol{X}_{i,1}, \ldots, \boldsymbol{X}_{i,t + 1}} = \lambda_{i, t + 1}\left(\frac{b + \sum_{k = 1}^{t} n_{i, k}}{a + \sum_{k = 1}^{t}\lambda_{i, k} - 1}\right).
    \end{align*}

Finally, as for the MVNB, the covariance between the number of claims of annual contracts of the same insured can be shown to be equal to: 
\begin{align*}
\Cov{N_{i,t}, N_{i,t+j}|\boldsymbol{X}_{i,1}, \ldots, \boldsymbol{X}_{i,t+j}} &= \lambda_{i,t} \lambda_{i,t+j} \left(\frac{b}{a-1}\right)
\left(\frac{b+1}{a-2} - \frac{b}{a-1} \right), \qquad j > 0.
\end{align*}
Obviously, for random effects models, other choices than the MVNB of the NBBeta are possible to construct.

\subsection{Dynamic panel data models}

At this point, it is worth mentioning one major drawback of the  use of a classic random effects model, such as the MVNB or the NB-Beta. By analyzing the predictive premiums, one can see that all past claims have the same importance, i.e., the same weight, in predicting the future premium. This means that a 10-year old claim is as important as a 1-year old one. In practice, it is generally accepted that this is not a realistic scenario: drivers evolve over time, and recent experience should have a greater impact than older experience when estimating a driver's risk.

Unfortunately, it is not easy to include such temporal dynamics in a panel data model. A random effects approach must assume a random process for $\Theta_{i,t}$. Thus, models where the random effects $\Theta_{i,t}$, $t = 1, \ldots, T_i$ evolve over time would need a $T$-dimensional integral to express the joint distribution of all claims of a single policyholder. Therefore, complex numerical procedures or approximated inference methods are needed (see \cite{Jung} for example). Other approaches have been proposed to include a dynamic effect into count models: evolutionary credibility models in \cite{Alb}, Poisson residuals in \cite{Pinquet}, or, more recently, copulas with the jittering method in \cite{Shi2014}.

In this paper, we consider the Harvey-Fernandes model, or H-F model, proposed by \cite{H1989} and introduced in actuarial science by \cite{B2007}. Let $\mathcal{H}_{i,t}$ denote the claim history up to time $t$ for contract $i$. This approach includes random effects that develop over time according to a two-step procedure:
\begin{itemize}
\item[\textbf{P-step}:] the conditional distribution of the random effect $\Theta_{i,t}|\mathcal{H}_{i,t}$ is \emph{predicted} and
\item[\textbf{U-step}:] the distribution is \emph{updated} according to some fonction $U$ and $\Theta_{i, t+1} \sim U(\Theta_{i,t}|\mathcal{H}_{i,t})$. 
\end{itemize}
If we select a conjugate distribution for the random effects (such as the gamma for the MVNB model, and the beta for the NBBeta model), $\Theta_{i,t}$ and $\Theta_{i,t}|\mathcal{H}_{i,t}$ should be from the same distribution. Thus, the selected function $U$ can be directly applied to the parameters of the distribution of $\Theta_{i,t}|\mathcal{H}{i,t}$ in order to obtain the distribution of $\Theta_{i,t+1}$ (\textbf{U-step}). Consistent with \cite{B2007}, we follow this path and select a function $U$ where the parameters $\alpha_t$, $\tau_t$ of the posterior distribution $\Theta_{i,t}|N_{i,t}$ will be modified as $\alpha_t^* = \nu \alpha_t$ and $\tau_t^* = \nu \tau_t$ with starting values $\alpha_0$ and $\tau_0$. From this structure we can derive the distribution of each $\Theta_{i,1}, \ldots, \Theta_{i,t}$.

Given that the joint distribution can be expressed as the product of predictive distributions, the following result can be shown (we remove $\boldsymbol{X}_{i}$ for simplicity):
    \begin{align*}
      &\Prr{N_{i,1} = n_{i,1}, \ldots, N_{i,t} =  n_{i,t}} =\\
& \ \Prr{N_{i,1} = n_{i,1}|\alpha_{i, 1}, \gamma_{i, 1}}  \Prr{N_{i,2} = n_{i,2}|\alpha_{i, 2}, \gamma_{i, 2}} \times \ldots \times \Prr{N_{i,t} = n_{i,t}| \alpha_{i, t}, \gamma_{i, t}},
    \end{align*}
    where only updated parameters are explicitly mentioned in each distribution:
\begin{align}\label{eq:rec}
  \alpha_{i, t} = \left(\nu\right)^{t-1}\alpha_{0} + \sum_{k=1}^{t-1}\left(\nu\right)^{k}n_{i, t-k}\nonumber \\
  \gamma_{i, t} = \left(\nu\right)^{t-1}\gamma_{0} + \sum_{k=1}^{t-1}\left(\nu\right)^{k}\lambda_{i, t-k}.
\end{align}
Thus, in line with \cite{B2007}, we can use a dynamic MVNB (noted HF-MVNB) by using the product of $\text{NB2}_{i,t}(\lambda_{i,t}, \alpha_{i,t}, \gamma_{i,t})$ distributions for $t=1,\ldots,T_i$.  Similarly, we can construct a new dynamic distribution based on the NBB distribution and then noted HF-NBB, as the product of  $\text{NBB}_{i,t}(\lambda_{i,t}, \alpha_{i,t}, \gamma_{i,t})$ distributions. We can see that both dynamic joint distribution put larger weight on recent claims in the predictive premium calculation:
     \begin{align*}
      \pi_{i, t+1}^{\text{HF-MVNB}} &= \Esp{N_{i,t+1}|\boldsymbol{N}_{i,t}, \boldsymbol{X}_{i,1}, \ldots, \boldsymbol{X}_{i,t + 1}}\\
       &= \lambda_{i, t + 1}\left(\frac{\left(\nu\right)^{t}\kappa + \sum_{k=1}^{t}\left(\nu\right)^{k}n_{i, t-k+1}}{\left(\nu\right)^{t}\kappa + \sum_{k=1}^{t}\left(\nu\right)^{k}\lambda_{i, t-k+1}}\right)
    \end{align*}
and
	\begin{align*}
      \pi_{i, t+1}^{\text{HF-NBB}} &= \Esp{N_{i,t+1}|\boldsymbol{N}_{i,t}, \boldsymbol{X}_{i,1}, \ldots, \boldsymbol{X}_{i,t + 1}}\\
       &= \lambda_{i, t + 1}\left(\frac{\left(\nu\right)^{t} b + \sum_{k=1}^{t}\left(\nu\right)^{k}n_{i, t-k+1}}{\left(\nu\right)^{t} a + \sum_{k=1}^{t}\left(\nu\right)^{k}\lambda_{i, t-k+1} - 1}\right).
        \end{align*}
        We discuss the covariance between the number of claims of annual contracts of the same policyholder in Section~\ref{sec:illustrations}.

 \section{Bonus-Malus Systems Panel Models}
 \label{ssec:BMS}

Bonus-malus systems (BMS) have been introduced in ratemaking procedures for many practical reasons (see \cite{Lemaire} and \cite{Denuit}). The basic idea underlying BMS models is to summarize past claims experience into a \textit{claim score}: a discrete value going from $1$ to $s$, where $1$ represent the lowest risk and $s$, the highest. In practice, when a new policyholder enters the portfolio, the insurance company gives that individual a selected entry level as an initial claim score.  Each year, depending on the insured's claim experience, the policyholder will move into the bonus-malus scale: toward high values for frequent claims, toward lower values in the opposite case. 

\begin{Definition}[Bonus-Malus System]\label{def:BMS} 
A Bonus-Malus System is a three-parameter model $(\Psi, s, \ell^*)$ where the level of the system at the beginning of period $t + 1$ in the database is given by
   \begin{align}\label{eq:LBMS}
     L(t+1) &= \min\left(\max\left(L(t) - \mathbb{I}\left(N_{i,t} = 0\right) + \Psi N_{i,t}, 1\right), s\right),
   \end{align}
   where $\mathbb{I}\left(N_{i,t} = 0\right)$ is a dummy variable indicating a period without claims, $s$ is the highest level of the system and $\Psi$ is the \emph{jump parameter}. The parameter $\ell^*$ corresponds to the entry level of the system for a policyholder without experience. This kind of structure of the BMS is commonly denoted by $-1/+\Psi$. \\
 \end{Definition}

Figure~\ref{fig:BMS1} presents an example of a bonus-malus system.\\
 
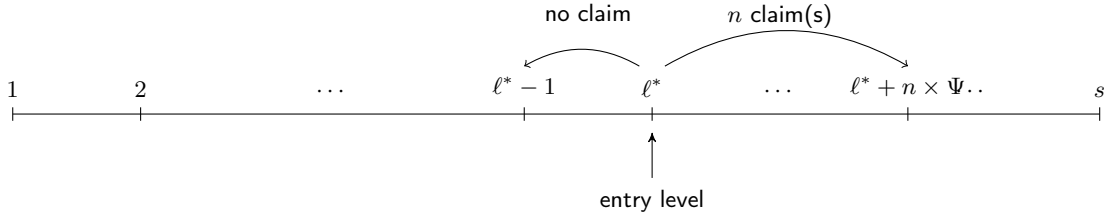
\begin{figure}[!htb]
\begin{center}
\begin{tikzpicture}[scale=0.85, auto, to/.style={->,>=stealth',shorten >=1pt}, every node/.style={font=\fontsize{9pt}{9pt}\selectfont\sffamily, align=center, semithick}]
\draw[-] (0,0) -- (17,0);
\foreach \x in {0, 2, 8, 10, 14, 17}   \draw (\x cm,3pt) -- (\x
cm,-3pt);
\draw (0,0) node[above=3pt] {$1$};
\draw (2,0) node[above=3pt] {$2$};
\draw (8,0) node[above=3pt] {$\ell^*-1$};
\draw (10,0) node[above=3pt] {$\ell^*$};
\draw (14,0) node[above=3pt] {$\ell^* + n \times \Psi$};
\draw (17,0) node[above=3pt] {$s$};
\draw (5,0) node[above=3pt] {$\cdots$};
\draw (12,0) node[above=3pt] {$\cdots$};
\draw (15,0) node[above=3pt] {$\cdots$};
\draw (10,-1) node[below=1pt] {entry level};
\draw (9,2) node[below=3pt] {no claim};
\draw (12,2) node[below=3pt] {$n$ claim(s)};

\draw[to] (10,-1) -- (10,-0.25);

\draw [->] (9.8,0.75) to [out=150,in=30] (8,0.75);
\draw [->] (10.2,0.75) to [out=30,in=150] (14,0.75);
\end{tikzpicture}
\caption{\textit{Example of a bonus-malus system. At the beginning of the first period, the new policyholder $i$ enters the system at level $\ell^*$. If no accidents are filed during the first insurance contract, the policyholder moves from level $\ell^*$ to level $\ell^*-1$ at the beginning of the second period.  For each claim, the policyholder's BMS level increases by $\Psi$.}}
\label{fig:BMS1}
\end{center}
\end{figure}

Intuitively, by averaging the claims frequency for each claim score, the insurance company could then construct a form of merit rating, representing relativities for BMS (see, for example, Figure 5 of \cite{Boucher2014}). However, we can improve the rating system by choosing other approaches than averaging the frequency of each claim score. An extensive body of literature deals with how to calibrate a BMS with cross-section insurance data (see \cite{Denuit} for an overview).  However, when we observe many annual contracts for a single insured in a data set (i.e. panel data structure), \cite{Boucher2014} presents a more suitable approach, where the authors construct what they call a BMS-panel model.

\subsection{BMS-panel model}

The main objective of the BMS-panel model is to estimate the parameters needed for the \textit{a priori} and the \textit{a posteriori} ratemaking simultaneously, by including the BMS structure directly in the modeling.  In order to estimate parameters, we need a model for the conditional random vector $\begin{bmatrix}\boldsymbol{N}_{i,t} | \boldsymbol{X}_{i,1}, \ldots, \boldsymbol{X}_{i,t}\end{bmatrix}$. First, let $\begin{bmatrix}N_{i,1}, N_{i,2}\end{bmatrix}$ be a random vector in a portfolio with a bonus-malus system as defined in Definition~\ref{def:BMS} and suppose that we know $\ell_{i,1}$, the BMS level of insured $i$ at time $1$.  The conditional joint probability mass function is given by (we drop the dependence on $\boldsymbol{X}$ to simplify the presentation)
 \begin{align*}
&\Pr(N_{i,1} = n_{i,1}, N_{i,2} = n_{i,2}|L(1) = \ell_{i,1}) \\
&= \Pr(N_{i,1} = n_{i,1}|L(1) = \ell_{i,1}) \Pr(N_{i,2} = n_{}i,2| N_{i,1} = n_{i,1}, L(1) = \ell_{i,1}) \\
   &= \Pr(N_{i,1} = n_{i,1}|L(1) = \ell_{i,1})\\
   &\phantom{=}\times\left( \sum_{y=1}^s \Pr(N_{i,2} = n_{i,2}| N_{i,1} = n_{i,1}, L(1) = \ell_{i,1}, L(2) = y) \Pr(L(2) = y| N_{i,1} = n_{i,1}, L(1) = \ell_{i,1})  \right). \intertext{Given the bonus-malus system, past information is captured by the last level reached by the system, we obtain}
&= \Pr(N_{i,1} = n_{i,1}|L(1) = \ell_{i,1}) \left( \sum_{y=1}^s \Pr(N_{i,2} = n_{i,2}| L(2) = y) \Pr(L(2) = y| N_{i,1} = n_{i,1}, L(1) = \ell_{i,1}) \right) \\
&= \Pr(N_{i,1} = n_{i,1}|L(1) = \ell_{i,1}) \Pr(N_{i,2} = n_{i,2}| L(2) = \ell_{i,2}), 
 \end{align*}
 where $\Pr(L(2) = y| N_{i,1} = n_{i,1}, L(1) = \ell_{i,1}) = 0$ for all $y$, except for $y= \ell_{i,2}$, the level reached by the bonus-malus system at time $2$ after $n_{i, 1}$ claims were observed during the year. Consequently,
\begin{align}\label{eq:ML}
\Pr(N_{i,1} = n_{i,1}, N_{i,2} = n_{i,2}, \ldots, N_{i,t} = n_{i,t}|L(1) = \ell_{i,1}) &= \prod_{k=1}^{t} \Pr(N_{i,k} = n_{i,k}|L(k) = \ell_{i,k}).
\end{align} 

Several options are available to model the conditional distribution $\Prr{N_{i,t} = n_{i,t}|L(t) = \ell_{i,t}}$: we can consider the distributions introduced in subsection~\ref{ssec:cross} or \ref{ssec:panel}, as well as the hurdle or the zero-inflated distributions (see \cite{Boucher2009} for an extensive overview of count distributions in ratemaking). 
For a selected distribution, the BMS level should be used to model the mean parameter of the conditional distribution.  Based on the Poisson, the NB2 or the NB1 introduced earlier, we assume that the mean parameter will be modeled as
\begin{align*}
   \pi_{i,t+1}^{\text{BMS}} &=  \Esp{N_{i,t+1}|\boldsymbol{N}_{i,t+1}, \boldsymbol{X}_{i,1}, \ldots, \boldsymbol{X}_{i, t+ 1}}= \lambda_{i, t + 1}r_{\ell_{i, t + 1}},
\end{align*}
where $\lambda_{i,t}$ is an \emph{a priori} premium based on the characteristics of the insured (sex, age, etc.) for the period $[t, t+1)$. Several structures for $r_{\ell}$ can be chosen, but we will restrict ourselves to studying linear relativities for the BMS, as defined below.
\begin{Definition}[Linear Bonus-Malus System]\label{def:LBMS} 
A linear relativity is associated with each step of the scale according to the equation
       \begin{align}\label{eq:penalty}
         r_{L(t)} &= 1 + \delta(L(t) - 1),
       \end{align}
       where $\delta$ is the \emph{penalty parameter}.
\end{Definition}

Equation~\eqref{eq:penalty} implies that $r_1 = 1$ and defines the basis risk. These linear relativities, initially proposed by \cite{Gilde}, prevent unwanted situations such as $r_i > r_j$ for $i < j$. 
 
\subsection{Entry level}

The distributions previously mentioned were defined conditionally on the knowledge of the BMS level of insured $i$ at time $1$.  Consequently, the joint probability mass function of policyholder $i$ at time $t$ is given by
\begin{align*}
  &\Prr{N_{i,1} = n_{i,1}, \ldots, N_{i,t} =  n_{i,t} | \boldsymbol{X}_{i,1}, \ldots, \boldsymbol{X}_{i,t}}\\
  &= \sum_{y = 1}^s \Prr{N_{i,1} = n_{i,1}, \ldots, N_{i,t} =  n_{i,t} | L(1) = y, \boldsymbol{X}_{i,1}, \ldots, \boldsymbol{X}_{i,t}}\Prr{L(1) = y}\\
  &= \sum_{y = 1}^s  \prod_{k=1}^{t} \Pr(N_{i,k} = n_{i,k}|L(k) = \ell_{i,k},\boldsymbol{X}_{i,1}, \ldots, \boldsymbol{X}_{i,k})\Prr{L(1) = y},
\end{align*}
where $\Prr{L(1) = y}$ is the probability distribution of the BMS level at time $t=1$, i.e, the first year an insured appears in the database. One must not confuse the first year of driving with the first year an insured is observed in the database, as illustrated in Figure~\ref{fig:BMS2}.

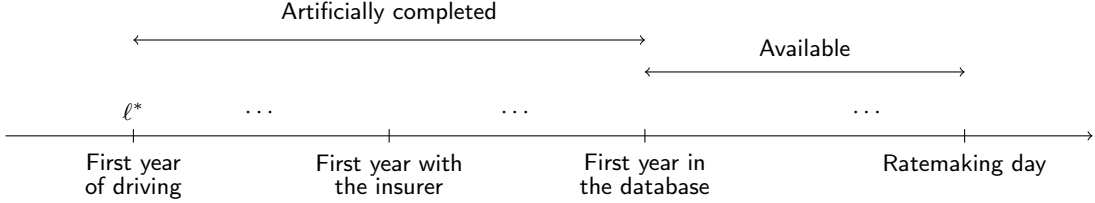
\begin{figure}[!htb]
\begin{center}
\begin{tikzpicture}[scale=0.85, auto, to/.style={->,>=stealth',shorten >=1pt}, every node/.style={font=\fontsize{9pt}{9pt}\selectfont\sffamily, align=center, semithick}]
  \draw[->] (0,0) -- (17,0);
  \draw[<->] (10,1) -- (15,1);
  \draw[<->] (2,1.5) -- (10,1.5);
\foreach \x in {2, 6, 10, 15}   \draw (\x cm,3pt) -- (\x
cm,-3pt);
\draw (2,0) node[below=3pt] {First year \\ of driving};
\draw (6,0) node[below=3pt] {First year with\\ the insurer };
\draw (10,0) node[below=3pt] {First year in\\ the database};
\draw (15,0) node[below=3pt] {Ratemaking day};
\draw (12.5,1.1) node[above=0pt] {Available};
\draw (6,1.6) node[above=0pt] {Artificially completed};
\draw (4,0) node[above=3pt] {$\cdots$};
\draw (8,0) node[above=3pt] {$\cdots$};
\draw (13.5,0) node[above=3pt] {$\cdots$};
\draw (2,0) node[above=2.5pt] {$\ell^*$};


\end{tikzpicture}
\caption{\textit{Example of a policyholder with partial unknown information. In order to estimate $\ell^*$, the insurer must complete the trajectory followed by the bonus-malus system between the first year of driving and the first year in the database.}}\label{fig:BMS2}
\end{center}
\end{figure}

If we want to model the joint distribution for a new driver, $\ell_{i,1}$ can be found directly: it is $\ell^*$, the entry level selected in the construction of the model. Thus 
$\Prr{L(1) = y} = 0$ for all $y$, except $y = \ell^*$. For experienced drivers, this situation is more difficult to deal with.  Insurers must be careful and should not suppose that all new insureds did not have claims in past years, nor must they automatically suppose that they should give  policyholders an entry level $\ell^*$.  In \cite{Boucher2014}, the authors propose to recreate all the possible events of each policyholder from the first year of driving to the first year in the database, to create the distribution of $L_1$.  By taking the average of $L_1$, 
they assigned each driver in the database a value of $\ell_1$. Although the method is intuitive, it is a very complex to implement in addition to requiring a lot of time and IT resources. In this paper, we propose a much simpler method based on the fact that the probability of not filing an accident during a single year is very high, often in the range of $80\% - 95\%$.

\begin{Proposition}\label{prop:JP}
  For a policyholder $i$, let $S_i(t)$ be the level reached by a $(\Psi, s, \ell^*)$ bonus-malus system after $t$ periods from the first exposure year and
$u_i$ be the first observed year in the portfolio. Thus, if
  \begin{align}\label{eq:condition}
    \Prr{N_{i, t} = 0| L(t) = \ell_{i, t}, \boldsymbol{X}_{i,1}, \ldots, \boldsymbol{X}_{i,t}} &> 0.5, \qquad \forall t,
  \end{align}
  then
  \begin{align*}
    \underset{\ell \in \{1, \ldots, s\}}{\text{argmax}}\left(\Prr{S_i(1) = \ell^*, \ldots, S_i(u_i+1) = \ell |\boldsymbol{X}_{i,1}, \ldots, \boldsymbol{X}_{i,t}}\right) &= \max\left(\ell^* - u_i, 1\right).
  \end{align*}
\end{Proposition}
\begin{proof}
Based on Equation~\eqref{eq:LBMS}, the transition probabilities are given by (for a policyholder with a claim frequency of $\lambda_{i, t}$ for the period $t$)
 \begin{align}\label{eq:trans}
   p_{k,j}(\lambda_{i, t}) &= \Prr{L(t+1) = j | L(t) = k, \boldsymbol{X}_{i, 1}, \ldots, \boldsymbol{X}_{i, t}} \nonumber\\
   &= \Prr{\min\left(\max\left(k - \mathbb{I}\left(N_{i, t} = 0\right) + \Psi N_{i, t}, 1\right), s\right) =  j| \boldsymbol{X}_{i, 1}, \ldots, \boldsymbol{X}_{i, t}},
 \end{align}
 for $j = 1, \ldots, s$ and $k = 1, \ldots, s$. This implies that $p_{k, \max\left(k - 1, 1\right)}(\lambda_{i, t}) > 0.5$ and $p_{k, j}(\lambda_{i,t}) < 0.5$, $j \ne  \max\left(k - 1, 1\right)$ if $\Prr{N_{i,t} = 0|L(t) = \ell_{i, t}, \boldsymbol{X}_{i, 1}, \ldots, \boldsymbol{X}_{i, t}} > 0.5$. Even if the insurance company does not know the  policyholders' characteristics before they enter the portfolio, it is reasonable to assume that they are such that Equation~\eqref{eq:condition} is satisfied.
 Thus,
 \begin{align*}
   &p_{\ell^*, \max\left(\ell^* - 1, 1\right)}(\lambda_{i, 1})p_{\max\left(\ell^* - 1, 1\right), \max\left(\ell^* - 2, 1\right)}(\lambda_{i, 2}) \times \ldots \times p_{\max\left(\ell^* - u_i + 1, 1\right), \max\left(\ell^* - u_i, 1\right)}(\lambda_{i, u_i})
 \end{align*}
 is the most likely path from level $\ell^*$ to level $\max\left(\ell^* - u_i, 1\right)$. Finally,
 \begin{align*}
   \underset{\ell \in \{1, \ldots, s\}}{\text{argmax}}
   \left(\Prr{S_i(1) = \ell^*, \ldots, S_i(u_i+1) = \ell} \right) &= \max\left(\ell^* - u_i, 1\right).
 \end{align*}
 \end{proof}

\begin{figure}[!htb]
\begin{center}
\begin{tikzpicture}[scale=0.85, auto, to/.style={->,>=stealth',shorten >=1pt}, every node/.style={font=\fontsize{9pt}{9pt}\selectfont\sffamily, align=center, semithick}]
  \draw[->] (0,0) -- (17,0);
  \draw[<->] (8,1.5) -- (15,1.5);
  \draw[<->] (2,2) -- (8,2);
\foreach \x in {2, 4, 8, 10, 12, 15}   \draw (\x cm,3pt) -- (\x
cm,-3pt);
\draw (2,0) node[below=3pt] {First year \\ of driving};
\draw (8,0) node[below=3pt] {First year in\\ the database};
\draw (15,0) node[below=3pt] {Ratemaking day};
\draw (11,1.6) node[above=0pt] {Available};
\draw (5,2.1) node[above=0pt] {Artificially completed};
\draw (6,0) node[above=3pt] {$\cdots$};
\draw (13.5,0) node[above=3pt] {$\cdots$};
\draw (2,0) node[above=2.5pt] {$\ell^*$};
\draw (8,0) node[above=2.5pt] {$L(1) = \ell_{i,1}$};
\draw (10,0) node[above=2.5pt] {$L(2) = \ell_{i,2}$};
\draw (12,0) node[above=2.5pt] {$L(3) = \ell_{i,3}$};
\draw (15,0) node[above=2.5pt] {$L(T_i) = \ell_{i,T_i}$};

\draw (2,0.5) node[above=2.5pt] {$S(1)$};
\draw (8,0.5) node[above=2.5pt] {$S(u_i + 1)$};
\draw (10,0.5) node[above=2.5pt] {$S(u_i + 2)$};
\draw (12,0.5) node[above=2.5pt] {$S(u_i + 3)$};
\draw (15,0.5) node[above=2.5pt] {$S(u_i + T_i)$};


\end{tikzpicture}
\caption{\textit{Example of a policyholder with partial unknown information ($u_i$ years are unknown and $T_i$ years are known).}}\label{fig:BMS22}
\end{center}
\end{figure}
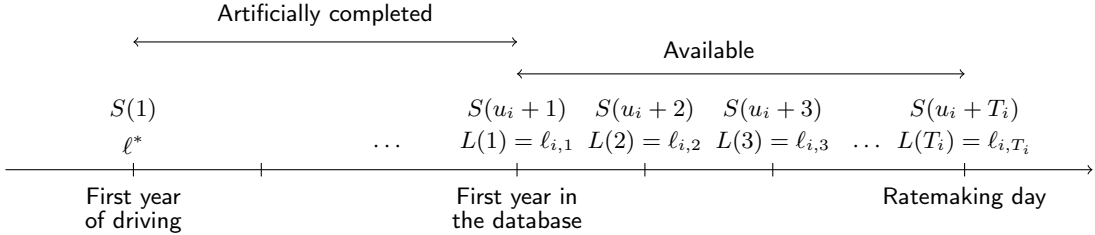

Figure \ref{fig:BMS22} illustrates the situation, and the new variables introduced. Based on the result of the last proposition, we then select the entry level of a new policyholder with $u$ years of experience as $\max\left(\ell^* -u, 1\right)$, or in other words $\Prr{L(1) = y} = 0$ for all $y$, except $y = \max\left(\ell^* -u, 1\right)$. Therefore, each year of driving experience results in a decrease of one level of the bonus-malus system, which is already one of the ways insurance companies deal with experienced drivers in merit rating plans. So, this way of selecting $\ell_{i,1}$ not only simplifies the BMS-panel model, but it 
theoretically justifies the procedure already in use.

\subsection{Parameters}

As the joint distribution of $N_{i,T_i}$ is now completely defined for all insured in the database. We can thus summarize the steps needed to estimate all the parameters of the model. 
First, the actuary has to select, or to estimate, the following structural components of the BMS-panel model:
\begin{itemize}
\item the number of levels $s$ of the system;
\item the jump parameter $\Psi$ for each claim;
\item the entry level $\ell^*$ for a new driver;
\item the risk characteristics $\boldsymbol{X}_{i,1}, \ldots, \boldsymbol{X}_{i,T_i}$ and the associated parameters $\boldsymbol{\beta}$ and 
\item the underlying count distribution.  
\end{itemize}

Then, for this specific BMS-panel model, $N_{i,T_i}$, the \emph{a priori} parameters ($\boldsymbol{\beta}$) linking the covariates with the expected value and the \emph{penalty parameter} $\delta$ must be estimated by maximizing the likelihood function. When structural components of the BMS-panel model are selected, parameter estimation is direct and can be done using standard statistical softwares. Other parameters linked to the underlying count distribution should also be estimated.  For example, if an NB1 or a NB2 distribution is used, an estimator of the overdispersion parameter $\alpha$ should also be found.

Note that for parameters $s$, $\Psi$ and $\ell^*$, the structure of the model provides us with some additional information that allows us to reduce the parameter space to a lattice. This is a well-known problem in statistics, e.g., see \cite{HAM} for simple cases or \cite{CHO} for asymptotic properties. To obtain the best BMS-panel model, we adjust all combinations of $s$, $\Psi$ and $\ell^*$ and we select the model that generates the best likelihood profile and/or the best prediction capacity (based on out-of-sample analysis).

\subsection{Properties of the BMS-model}
\label{ssec:cov}

Amoung the properties of the BMS-panel model, we are interested in evaluating the covariance between premiums paid by a policyholder for two periods. To do this, it is necessary to study the covariance between $N_t$ and $N_{t+k}$, $k = 1, 2, \ldots$, conditional to $L(t) = \ell_t$.\\

For a BMS-panel model, the one-year probability of the random variable $L$ going from BMS level $\ell_{i, t}$ to BMS level $\ell_{i, t+1}$ is denoted by $p_{\ell_{i,t}, \ell_{i, t+1}}(\lambda_{i, t})$ as defined by Equation~\eqref{eq:trans}. For a policyholder, we can construct a transition probability matrix
\begin{align*} 
\boldsymbol{P}(\lambda_{i,t})  &=  
\begin{bmatrix}
p_{1,1}(\lambda_{i,t}) &  p_{1,2}(\lambda_{i,t}) & \cdots & p_{1,s}(\lambda_{i,t}) \\ 
p_{2,1}(\lambda_{i,t}) &  p_{2,2}(\lambda_{i,t}) & \cdots & p_{2,s}(\lambda_{i,t}) \\ 
\vdots &  \vdots & \ddots & \vdots \\ 
p_{s,1}(\lambda_{i,t}) &  p_{s,2}(\lambda_{i,t}) & \cdots & p_{s,s}(\lambda_{i,t}) \\ 
\end{bmatrix}.  
\end{align*}
We can show that for all $K = 1, 2, \ldots$, we have
\begin{align*} 
\boldsymbol{P}^{(K)}(\lambda_{i, t})  &= \boldsymbol{P}^{K}(\lambda_{i,t}),
\end{align*}
meaning that the transition probability matrix over $K$ time periods is simply the $K^{\text{th}}$ power of the annual transition probability matrix $\boldsymbol{P}(\lambda_{i,t})$.

\begin{Proposition}\label{prop:cov}
  In a BMS-panel model, at the beginning of a period $t$, the conditional covariance between $N_{i,t}$ and $N_{i, t+j}$, $j = 1, 2, \ldots$ is 
  \begin{align*}
    &\Cov{N_{i, t}, N_{i, t+j}|\ell_{i, t}, \boldsymbol{X}_{i, 1}, \ldots, \boldsymbol{X}_{i, t}} = \lambda_{i, t+j} \\
    &\phantom{=}\times \sum_{m=1}^s r_m\left(\Esp{N_{i,t} p_{\min\left(\max\left(\ell_{i,t} - \mathbb{I}(N_{i,t} = 0) + \Psi N_{i,t}, 1\right), s\right),m}^{(j-1)}\left(\lambda_{i,t}\right)|\ell_{i,t}, \boldsymbol{X}_{i, 1}, \ldots, \boldsymbol{X}_{i, t}} - \lambda_{i, t}r_{\ell_{i,t}}p_{\ell_{i,t}, m}^{(j)}\left(\lambda_{i,t}\right)\right),
  \end{align*}
  where the transition probabilities are extracted from the transition matrices $\boldsymbol{P}(\lambda_{i,t})^{(j)}$ and $\boldsymbol{P}(\lambda_{i,t})^{(j-1)}$ assuming that $\boldsymbol{X}_{i, t} = \boldsymbol{X}_{i, t+1} = \boldsymbol{X}_{i, t+2} = \cdots$.
\end{Proposition}
\begin{proof}
  We have (we drop the reference to $i$ and to $\boldsymbol{X}_{i, 1}$, $\ldots$, $\boldsymbol{X}_{i, t}$ in order to ease the presentation)
\begin{align*}
\Esp{N_t|\ell_t} &= \lambda_t r_{\ell_t}\\
\Esp{N_{t+j}|\ell_t} &= \sum_{n = 0}^{\infty} n \Prr{N_{t+j} = n|\ell_t} \\
                 &= \sum_{m = 0}^{\infty}  \sum_{n=0}^{\infty} n \Prr{N_{t+j} = n|L(t+j) = m, \ell_t} \Prr{L(t+j) = m|\ell_t}. \intertext{The inner sum is}
                   \sum_{n=0}^\infty n\Prr{N_{t+j} = n | L(t+j) = m, \ell_t}
                 &= \sum_{n=0}^\infty n\Prr{N_{t+j} = n |  L(t+j) = m}\\
                 &= \Esp{N_{t+j}|  L(t+j) = m}\\
                 &= \lambda_{t+j}r_m.
\end{align*}
Thus,
\begin{align*}
 	  	\Esp{N_{t+j}|\ell_t} &= \sum_{m=0}^{\infty}\lambda_{t + j}r_m \Prr{L(t+j) = m|\ell_t}\\
                                     &= \sum_{m=0}^{\infty}\lambda_{t+j}r_m\sum_{q_1 = 0}^\infty  \Prr{L(t+j) = m| L(t+j-1) = q_1, \ell_t} \Prr{L(t+j-1) = q_1 | \ell_t}\intertext{and, assuming that $\boldsymbol{X}_{i, t} = \boldsymbol{X}_{i, t+1} = \boldsymbol{X}_{i, t+2} = \cdots$, we have}
                                     &= \lambda_{t+j}\sum_{m=0}^\infty r_m\sum_{q_1 = 0}^\infty p_{q_1, m}(\lambda_{i,t}) \Prr{L(t+j-1)= q_1 | \ell_t}. \intertext{Recursively, we obtain}
                                     &= \lambda_{t+j}\sum_{m=0}^\infty r_m\sum_{q_1 = 0}^\infty p_{q_1, m}(\lambda_{i,t})  \cdots \sum_{q_{j-1}=0}^\infty p_{q_{j-1}, q_{j-2}}(\lambda_{i,t}) p_{\ell_t, q_{j-1}}(\lambda_{i,t}) , \intertext{and}
   &= \lambda_{t+j}\sum_{m=0}^\infty r_mp_{\ell_t, m}^{(j)}(\lambda_{i,t}). 
\end{align*}
\begin{align*}
\Esp{N_t N_{t+j}|\ell_t} &= \sum_{n_0 = 0}^{\infty} \sum_{n_j = 0}^{\infty} n_0n_j \Prr{N_t = n_0, N_{t+j} = n_j|\ell_t} \\
 &= \sum_{n_0 = 0}^{\infty} \sum_{n_j = 0}^{\infty} n_0 n_j \Prr{N_t = n_0|\ell_t} \Prr{N_{t+j} = n_j|\ell_t, N_t = n_0} \\
                         &= \sum_{n_0 = 0}^{\infty}\sum_{n_j=0}^\infty n_0n_j \Prr{N_t = n_0|\ell_t} \sum_{m = 0}^{\infty}  \Prr{N_{t+j} = n_j|L(t+j) = m, \ell_t, N_t = n_0}\\
  &\phantom{=}\times\Prr{L(t+j) = m |\ell_t, N_t = n_0} \\
                         &= \sum_{n_0=0}^{\infty}\sum_{n_j=0}^\infty\sum_{m=0}^\infty n_0 n_j \Prr{N_t = n_0|\ell_t} \Prr{N_{t+j} = n_j|L(t+j) = m}\\
  &\phantom{=}\times \sum_{q_1=0}^\infty \Prr{L(t+j) = m|  L(t+j-1) = q_1}\Prr{L(t+j-1) = q_1| \ell_t, N_t=n_0}.\intertext{Recursively, we obtain}
                         &= \sum_{n_0=0}^{\infty}\sum_{n_j=0}^\infty\sum_{m=0}^\infty n_0 n_j \Prr{N_t = n_0|\ell_t}\Prr{N_{t+j} = n_j  | L(t+j) = m}\\
  &\phantom{=} \times p_{\min\left(\max\left(\ell_t - \mathbb{I}(n_0 = 0) + \Psi n_0, 1\right), s\right), m}^{(j-1)}\left(\lambda_{i,t}\right)\\
                         &= \lambda_{t+j} \sum_{n_0=0}^{\infty} \sum_{m=0}^\infty r_m n_0\Prr{N_t = n_0 |  \ell_t}p_{\min\left(\max\left(\ell_t - \mathbb{I}(n_0 = 0) + \Psi n_0, 1\right), s\right), m}^{(j-1)}\left(\lambda_{i,t}\right)\\
  &= \lambda_{t+j} \sum_{m=0}^\infty r_m \Esp{N_t p_{\min\left(\max\left(\ell_t - \mathbb{I}(N_t = 0) + \Psi N_t, 1\right), s\right), m}^{(j-1)}\left(\lambda_{i,t}\right)|  \ell_t}
\end{align*}
and the result follows directly from the definition of a covariance.
\end{proof}

For $j = 1$, i.e., for two successive periods, the result of the Proposition~\ref{prop:cov} simplifies to
\begin{align*}
  \Cov{N_{i,t}, N_{i,t+1}|\ell_{i,t}, \boldsymbol{X}_{i,1}, \ldots, \boldsymbol{X}_{i, t}} &= \sum_{n=0}^\infty \sum_{q=0}^\infty nq\Prr{N_{i,t} = n|\ell_{i,t}, \boldsymbol{X}_{i,1}, \ldots, \boldsymbol{X}_{i, t}}\\
  &\phantom{=}\times \Prr{N_{i, t+1} = q|\ell_{i,t+1}, \boldsymbol{X}_{i,1}, \ldots, \boldsymbol{X}_{i, t}}\\
  &\phantom{=}- \lambda_{i,t}\lambda_{i,t+1}r_{\ell_{i,t}}\sum_{m=1}^s r_mp_{\ell_{i,t}, m}(\lambda_{i,t}),
\end{align*}
with $\ell_{i, t+1} = \ell_{i,t} - \mathbb{I}(n = 0) + n\Psi$.

\section{Empirical Illustrations}
\label{sec:illustrations}

We analyze a database from a portfolio of general liability insurance policies for private individuals from a major Canadian insurance company. Because bonus-malus systems are generally built for auto insurance products, we only consider the private use of vehicles in our project. The sample contains information about $\numprint{429333}$~contracts from $\numprint{140714}$~policyholders, and runs from year~$2012$ to year~$2016$. The sample has the following properties:
\begin{enumerate}
\item We only keep policyholders with a maximum of one insured car through all their observed contracts in order to avoid possible within-contract dependence;
\item We only keep policyholders with complete coverages on all their observed contracts, to avoid censorship where some claims would not be covered or observed;
\item We remove from the database policyholders who use their car for commercial purpose since they might exhibit different driving behaviors.   
\end{enumerate}

These selection criteria might affect the results.  For example, by choosing only insureds with one car, we might have a higher proportion of single insured than normally.  However, the purpose of this study is not to explain the accident process by covariates 
nor the claiming process, but only to present interesting count models that can be considered for ratemaking. 
To evaluate the performance our models, we split the database into two components: a fitting set ($\numprint{98589}$~policyholders) and a validation set 
($\numprint{42125}$~policyholders). Table~\ref{tab:desc} describes the $8$~covariates selected in the modeling. For every contract, we have initial information at the beginning of the period and we are interested in predicting the number of claims (excluding comprehensive claims). The average claim frequency is approximately $6.5\%$ and we observe a maximum of $5$~claims per contract. The distribution of the number of contracts observed per insured is shown in Table~\ref{tab:pctan}, for an average of $3.05$~contracts per policyholder.

\begin{table}
  \centering
  \begin{tabular}{ll}
    \toprule
    Variable & Description\\
    \midrule
    $X_1 = 1$ & The policyholder is female.\\
    $X_2 = 1$ & The policyholder is married.\\
    $X_3 = 1$ & The policyholder is less than $30$~years old.\\
    $X_4 = 1$ & The policyholder is between $30$ and $50$~years old.\\
    $X_5 = 1$ & The distance driven per year is less than $\numprint{10000}$~km.\\
    $X_6 = 1$ & The distance driven per year is between $\numprint{10000}$~km and $\numprint{20000}$~km.\\
    $X_7 = 1$ & The distance driven per year is between $\numprint{20000}$~km and $\numprint{30000}$~km.\\
    $X_8 = 1$ & The vehicule is used to commute.\\
    \bottomrule
  \end{tabular}
  \caption{Dichotomous variables in the database.}
  \label{tab:desc}
\end{table}

\begin{table}
  \centering
  \begin{tabular}{ll}
    \toprule
    Number of year & Percentage\\
    \midrule
$1$  & $21.42\%$ \\
$2$  & $17.73\%$ \\
$3$  & $11.66\%$ \\
$4$  & $32.57\%$ \\
$5$  & $16.62\%$ \\
    \bottomrule
  \end{tabular}
  \caption{Distribution of the number of contracts}
  \label{tab:pctan}
\end{table}

\subsection{Count distributions}
\label{ssec:RESBMS}

We present, in Table~\ref{tab:res1}, fitting results for various models. We observe that the Harvey-Fernandes modification of both the MVNB and the NBBeta distribution did not improve the fitting, resulting in a value of $\nu=1$ for both models.  We suppose that this is explained by the average of $3.05$~contracts per policyholder, where the weight of past claims can still be supposed equal in the prediction of new claims.

\begin{table}
  \centering
  \begin{tabular}{lllll}
    \toprule
    Model & $\#$ parameters & Logl. & AIC & BIC \\
    \midrule
Poisson			&$9$	&$\numprint{-85526,64}$	&$\numprint{171071.28}$	&$\numprint{171280.34}$ \\
NB2			&$10$	&$\numprint{-85238.04}$	&$\numprint{170496.08}$	&$\numprint{170728.37}$\\
NB1			&$10$	&$\numprint{-85149.08}$	&$\numprint{170318.16}$	&$\numprint{170550.45}$\\
MVNB			&$10$	&$\numprint{-84812.54}$	&$\numprint{169645.08}$	&$\numprint{169877.37}$\\
NBBeta			&$11$	&$\numprint{-84727.13}$	&$\numprint{169476.26}$	&$\numprint{169731.78}$\\
HF-MVNB			&$11$	&$\numprint{-84812.54}$	&$\numprint{169647.08}$	&$\numprint{169902.60}$\\
HF-NBBeta		&$12$	&$\numprint{-84727.13}$	&$\numprint{169478.26}$	&$\numprint{169757.01}$\\ \hline
MVNB$^*$		&$10$	&$\numprint{-84611.88}$	&$\numprint{169243.76}$	&$\numprint{169476.05}$\\
NBBeta$^*$		&$11$	&$\numprint{-84426.63}$	&$\numprint{168875.26}$	&$\numprint{169130.78}$\\
HF-MVNB$^*$		&$11$	&$\numprint{-84460.82}$	&$\numprint{168943.64}$	&$\numprint{169199.16}$\\
HF-NBBeta$^*$	        &$12$	&$\mathbf{\numprint{-84355.07}}$	&$\mathbf{\numprint{168734.14}}$	&$\mathbf{\numprint{169012.89}}$\\
    \bottomrule
  \end{tabular}
  \caption{Fitting results}
  \label{tab:res1}
\end{table}

We can observe two results for each panel data count distribution (with and without $^*$). This can be explained by additional informations in our database. Indeed, the province of Ontario uses Autoplus, a database that provides detailed automobile claims and policy history. Thus, it is possible for insurers to have access to the past claims history of a policyholder up to $10$~years.  In the construction of the database used for this project, we are then able to know the past 10 years of claim experience for each insured (before their first insurance contract observed in the database). However, only the claim experience for each insured is available, and the risk characteristics ($\boldsymbol{X})$ for those $10$~past years are unknown.  Because we will use this information in the BMS-panel model, we also tried to include this past experience in panel data models. Thus, we propose  a slightly modified version of the MVNB model and the NBBeta model, which we note MVNB$^*$ and NBBeta$^*$. Specifically, we assume that the distribution of the random variable $\Theta_i$ has already been adapted to account for past claims. Following a similar development to that described above,
we could show that for a policyholder with $m$~years of experience before entering the database, the distribution of $\Theta_i$ for the first observed contract of an insured
(gamma for the MVNB, beta for the NBBeta) will have the following parameters
    \begin{align*}
      \alpha^* &= \alpha + \sum_{j=1}^{m_i} n_{i,-j}^*   \quad \text{ and } \quad \gamma^* = \gamma + \sum_{j=1}^{m_i} \lambda_{i,-j},
    \end{align*}
where $m_i$ is the minimum value between $10$ and the driving experience (in years) for insured $i$, $n_{i,-t}^*$ is the observed number of claim(s) $t$~year(s) before the entry of the policyholder into the database. Because we do not observe the $\lambda_{i,-j}$, $j = 1, \ldots, m_i$, we approximate them by $\overline{\lambda} = 6.5\%$, the average frequency of the database.\\
 
Finally, the predicted premium for the MVNB$^*$ can be shown to be equal to
\begin{align}\label{eq:predMVNB}
\pi_{i, t + 1}^{\text{MVNB}^*} &= \lambda_{i,t + 1} 
\left(\frac{\sum_{k=1}^{t} n_{i,k} + \alpha^*}{\sum_{k=1}^{t}\lambda_{i,k} 
+ \gamma^*}\right) \nonumber \\
&=  \lambda_{i,t+ 1} 
\left(\frac{\sum_{j=1}^{m_i} n_{i,-j}^* + \sum_{k=1}^{t} n_{i,k} + \alpha}{m_i \lambda + \sum_{k=1}^{t} \lambda_{i,k} + \gamma}\right),
\end{align}
where all the available claim experience can be used to estimate the future premium.  The predicted premium for the NB-Beta$^*$ can also be computed straightforwardly.  Note that modified Harvey-Fernandes approaches (HF-MVNB$^*$ and HF-NBBeta$^*$) can also be constructed using the same procedure, with 

\begin{align}\label{eq:rec}
  \alpha_{i, 1}^* = \left(\nu\right)^{m_i}\alpha_{0} + \sum_{k=1}^{m_i}\left(\nu\right)^{k}n_{i, -k}\nonumber \intertext{and}
  \gamma_{i, 1}^* = \left(\nu\right)^{m_i}\gamma_{0} + \sum_{k=1}^{m_i}\left(\nu\right)^{k}\lambda_{i, -k}.
\end{align}

When we include this new information, the HF-NBBeta$^*$ outperforms all the other distributions in fitting statistics, even if we considerer penalized criteria such as the AIC or the BIC. We also use the out-of-sample data to compare the models.  Table~\ref{tab:res2} shows the results where we compute two measures evaluating the prediction capacity of each model: the mean-square error (MSE), and~$-$~because we are dealing with count data and not a continuous distribution~$-$~a loglikehood statistic from the Poisson distribution. Based on both out-of-sample statistics, the HF-MVNB$^*$ seems to offer a better prediction capacity.
\begin{table}
  \centering
  \begin{tabular}{lll}
    \toprule
    Model & Logl.(Poisson)  & MSE \\
    \midrule
Poisson			& $\numprint{-36966.48}$ &$\numprint{10609.54}$ \\ 
NB2 			& $\numprint{-36860.55}$ &$\numprint{10611.03}$ \\
NB1 			& $\numprint{-36815.96}$ &$\numprint{10609.53}$ \\
MVNB			& $\numprint{-36731.55}$ &$\numprint{10567.37}$ \\
NB-Beta 		& $\numprint{-36830.10}$ &$\numprint{10601.31}$ \\
HF-MVNB 		& $\numprint{-36731.55}$ &$\numprint{10567.37}$ \\
HF-NBBeta 		& $\numprint{-36830.10}$ &$\numprint{10601.31}$ \\ \hline
MVNB$^*$		& $\numprint{-36588.47}$ &$\numprint{10552.34}$ \\
NB-Beta$^*$		& $\numprint{-36628.14}$ &$\numprint{10564.14}$ \\
HF-MVNB$^*$		& $\mathbf{\numprint{-36543.69}}$ & $\mathbf{\numprint{10543.84}}$ \\
HF-NBBeta$^*$	        & $\numprint{-36628.04}$ &$\numprint{10570.82}$ \\
    \bottomrule
  \end{tabular}
  \caption{Out-of-sample statistics.}
  \label{tab:res2}
\end{table}

\subsection{Bonus-Malus systems panel models}
\label{ssec:RESBMS}

We adjust a bonus-malus system panel model with Poisson, NB1 and NB2 underlying distributions. As we did for the previous models marked with $^*$, we also consider past claim experience from \textit{Autoplus}.  This $10$-year history of past claims allows us to find $\ell_1$, the BMS level of each insured when they are observed in the database at the first time. We choose several combinations of structure parameters on a grid given by $s = 2, 3, \ldots, S$, $\Psi = 1, 2, \ldots, s$ and $\ell^*=1,...,s$, where $S = 22$ for the Poisson distribution, $S = 16$ for the NB1 distribution and $S = 15$ for the NB2 distribution.  These values of $S$ cover between $\numprint{1000}$ and $\numprint{3000}$~possibilities for each underlying distribution. Because each estimation step takes at least $2$ to $5$~minutes on a personal computer, covering all those possibilities is very time consuming.  We are looking for a procedure that could help us to narrow the space of $\{s, \Psi, \ell^*\}$ to find the best combinations for a specific underlying distribution.
\begin{table}
  \centering
  \begin{tabular}{llllllll}
    \toprule
Distribution & \# par. & $\Psi$ &	$s$ &	$\ell^*$ & Loglikelihood  & AIC &BIC \\		
    \midrule
Poisson &$13$&$6$ &$11$ &$1$ &$\numprint{-84672.04}$ &$\numprint{169370.08}$ &$\numprint{169508.07}$ \\ 
		&$13$&$6$ &$11$ &$2$ &$\numprint{-84672.63}$ &$\numprint{169371.26}$ &$\numprint{169509.25}$ \\
		&$13$&$6$ &$10$ &$1$ &$\numprint{-84673.27}$ &$\numprint{169372.54}$ &$\numprint{169510.53}$ \\
		&$13$&$6$ &$10$ &$2$ &$\numprint{-84673.89}$ &$\numprint{169373.78}$ &$\numprint{169511.77}$ \\ \hline
NB1		&$14$&$6$ &$11$ &$1$ &$\mathbf{\numprint{-84331.55}}$ &$\mathbf{\numprint{168691.11}}$ &$\mathbf{\numprint{168839.71}}$ \\ 
		&$14$&$6$ &$10$ &$1$ &$\numprint{-84331.92}$ &$\numprint{168691.84}$ &$\numprint{168840.45}$ \\
		&$14$&$6$ &$11$ &$2$ &$\numprint{-84332.13}$ &$\numprint{168692.27}$ &$\numprint{168840.87}$ \\
		&$14$&$6$ &$10$ &$2$ &$\numprint{-84332.52}$ &$\numprint{168693.05}$ &$\numprint{168841.65}$ \\ \hline
NB2		&$14$&$6$ &$11$ &$1$ &$\numprint{-84442.99}$ &$\numprint{168913.97}$ &$\numprint{169062.58}$ \\
		&$14$&$6$ &$11$ &$2$ &$\numprint{-84443.56}$ &$\numprint{168915.13}$ &$\numprint{169063.73}$ \\
		&$14$&$6$ &$10$ &$1$ &$\numprint{-84444.49}$ &$\numprint{168916.98}$ &$\numprint{169065.58}$ \\
		&$14$&$6$ &$10$ &$2$ &$\numprint{-84445.09}$ &$\numprint{168918.19}$ &$\numprint{169066.79}$ \\ 
    \bottomrule
  \end{tabular}
  \caption{BMS-panel data statistics.}
  \label{tab:res3}
\end{table}
 
For each model, we estimate $\boldsymbol{\beta}$, $\delta$ and an overdispersion parameter (for both NB distributions) using the estimation set. Finally, as we did for the other distributions, we calculate the value of the mean square error, as well as the loglikelihood of a Poisson distribution for each fitted model based on the test sample, in both cases to prevent over-adjustment. In Table~\ref{tab:res3}, we present the results of the four best SBM-panel models for all three underlying distributions, evaluated on the estimation dataset. We observe that the model $-1/+6$ with $s=11$ is the best model for the Poisson, the NB1 and the NB2.  The four best models for each underlying distribution have the same structural parameters.  Only the second best model of the NB1 ranks third for the Poisson and the NB2.  The value of $\delta$ does not seem to depend on the underlying distribution but rather depends on the $3$~structural parameters $s, \Psi$ and $\ell^*$.  Finally, we can see that the underlying distribution NB1 offers better fitting statistics than the NB2, reflecting what we already observed for cross-section data distributions. Formally, some statistical tests could be done to determine if the overdispersion parameter of the NB1 and the NB2 is statistically significant, but the huge differences between loglikelihoods already show that is the case. As done previously, we also use compute out-of-sample statistics to assess the prediction capacity of BMS-models.  Table~\ref{tab:res33} presents the results.
\begin{table}
  \centering
  \begin{tabular}{llllll}
    \toprule
Distribution & $s$ &	$\Psi$ &	$\ell^*$ & MSE & Logl.(Poisson) \\		
    \midrule
Poisson 
&$11$	&$6$	&$1$	&$\numprint{10550.35}$	&$\numprint{-36587.65}$\\
&$11$	&$6$	&$2$	&$\numprint{10550.41}$	&$\numprint{-36587.94}$\\
&$10$	&$6$	&$1$	&$\numprint{10549.90}$	&$\numprint{-36586.96}$\\
&$10$	&$6$	&$2$	&$\numprint{10549.91}$	&$\numprint{-36587.22}$\\ \hline
NB1 
&$11$	&$6$	&$1$	&$\numprint{10550.41}$	&$\numprint{-36587.57}$\\
&$10$	&$6$	&$1$	&$\numprint{10549.93}$	&$\numprint{-36586.85}$\\
&$11$	&$6$	&$2$	&$\numprint{10550.42}$	&$\numprint{-36587.80}$\\
&$10$	&$6$	&$2$	&$\numprint{10549.95}$	&$\numprint{-36587.11}$\\ \hline
NB2 
&$11$	&$6$	&$1$	&$\numprint{10552.06}$	&$\numprint{-36586.44}$\\
&$11$	&$6$	&$2$	&$\numprint{10552.14}$	&$\numprint{-36586.74}$\\
&$10$	&$6$	&$1$	&$\numprint{10551.54}$	&$\numprint{-36585.73}$\\
&$10$	&$6$	&$2$	&$\numprint{10551.61}$	&$\numprint{-36586.04}$\\
    \bottomrule
  \end{tabular}
  \caption{Out-of-sample statistics for the BMS-panel model.}
  \label{tab:res33}
\end{table} 

The best model selected by the estimation dataset is the NB1, $-1/+6$ with $s=11$ and $\ell^* = 1$.  Out-of-sample statistics show that this model has an interesting predictive capacity, only outperformed by the HF-MVNB$^*$.  We also see that the BMS-panel data model 
makes stable predictions, whereas the $15$ other BMS-panel models selected show similar out-of-sample statistics (the greatest difference is less than $0.005\%$). 

Finally, because the best model selected by the data is the NB1, $-1/+6$ with $s=11$ and $\ell^* = 1$, we point out that
\begin{itemize}
	\item a policyholder reaches the top of the BMS after $2$~claims in $3$~years, resulting in an annual premium $2.2$~times higher than the premium of a policyholder located at level~$1$;
	\item after one claim, it takes $6$~years for an insured 
	revert to their initial position; and 
	\item new drivers are assigned an entry level of $1$, the best of the BMS.  This result highly depends on the data used, but is in total opposition to what \cite{Boucher2014} obtained.  Indeed, these authors conclued that new drivers should be assigned to the worst BMS scale.
\end{itemize}

\subsection{Parameters Analysis}

Table~\ref{tab:resbeta} presents estimated values and standard errors for all $\beta$ parameters, for some of the best models namely the NB1 and the HF-NBB models, and the BMS-NB1 model with $s=11$, $\ell^* = 1$ and $\Psi = 6$. For the latter, we do not consider the variability of the structural parameters $s, \ell^*$ and $\Psi$ in the analysis. As already noticed by \cite{Boucher2014}, it is interesting to note the difference between the $\widehat{\beta}$, particularly when comparing cross-section data model (such as the NB1) and the other two models that allow for merit rating.  Estimated parameters $\widehat{\beta_5}$, $\widehat{\beta_6}$ and $\widehat{\beta_7}$ associated to the distance driven, show the largest difference between models. Smaller differences between those $\widehat{\beta}$ can be observed between the HF-NBB model and the BMS-NB1 model. It means that the form of the premium penalties for claiming has an impact of the \textit{a priori} rating. 

\begin{table}
	\centering
	\begin{tabular}{lrrrrrr}
 		\toprule
 &		\multicolumn{ 2}{c}{NB1} &	\multicolumn{ 2}{c}{HF-NBB} & \multicolumn{ 2}{c}{BMS-NB1} \\	
Parameters &Est. & Std. err. &Est. & Std. err. &Est. & Std. err.  \\		
		\midrule
$\widehat{\beta_0}$ &$-2.103$ & ($0.036$)	&$1.527$&   ($0.106$)	&$-2.356$&	($0.031$) \\
$\widehat{\beta_1}$ & $0.037$	 & ($0.010$)	&$0.037$&	  ($0.015$)	&$0.031$&	    ($0.013$) \\
$\widehat{\beta_2}$ &$-0.026$ & ($0.014$)	&$-0.036$&  ($0.015$)	&$-0.031$&	($0.014$) \\
$\widehat{\beta_3}$ &$0.424$	 & ($0.018$)	&$0.417$&	  ($0.019$)	&$0.403$&	    ($0.018$) \\
$\widehat{\beta_4}$ & $0.345$	 & ($0.015$)	&$0.323$&	  ($0.016$)	&$0.309$&	    ($0.016$) \\
$\widehat{\beta_5}$ &$-0.579$ & ($0.041$)	&$-0.497$&  ($0.071$)	&$-0.481$&	($0.028$) \\
$\widehat{\beta_6}$ &$0.453$ & ($0.042$)	&$-0.397$&  ($0.070$)	&$-0.385$&	($0.032$) \\
$\widehat{\beta_7}$ &$-0.245$ & ($0.050$)	&$-0.222$&  ($0.074$)	&$-0.216$&	($0.039$) \\
$\widehat{\beta_8}$ & $0.029$ & ($0.013$)	& $0.047$&  ($0.016$)	& $0.035$&    ($0.016$) \\
     \bottomrule
\end{tabular}
\caption{Estimated $\beta$ (std.err.) parameters for specified models}
\label{tab:resbeta}
\end{table} 

\subsection{Predictive and Covariance Analysis}

\begin{table}
	\centering
	\begin{tabular}{lrrr}
		\toprule
		Models &	Parameter &	Estimation & Std. Err \\		
		\midrule
NB1       & $\tau$ &   $0.067$  &  ($0.004$) \\ \midrule
HF-NBB    & $a$      & $264.818$  & ($12.881$) \\
          & $b$      &   $5.500$  &  ($0.471$) \\
          & $\nu$    &   $0.900$  &  ($0.008$) \\ \midrule
BMS-NB1   & $\tau$ &   $0.062$  &  ($0.003$) \\
          & $\delta$ &   $0.120$  &  ($0.004$) \\
     \bottomrule
\end{tabular}
\caption{Estimated parameters for specified models}
\label{tab:resstruct}
\end{table} 

Table~\ref{tab:resstruct} shows the other estimated parameters for the same models.  For the NB1 distribution, the value of $\tau$ mainly measures the overdispersion of the count distribution and does not model the dependence between number of claims.  Indeed, the cross-section models suppose independence between all contracts and a null covariance between $N_{i,t}$ and $N_{i,t+j}$. This means that no merit rating is possible for this class of models. By comparison, for the more classic panel data models, such as the MVNB or the NBBeta, we established that the covariance between $N_{i,t}$ and $N_{i,t+j}$ is constant and does not depend on $j$. This results in the fact that the age of a claim is not being considered in the merit rating plan scheme.

The HF-NBB models, generalized quite directly by something similar to a Kalman filter approach, were designed to allow different weights depending on the age of the claim.  Analytic solution to compute the covariance is too complex, so we simply simulate the values. Figure~\ref{fig:covHF} shows two graphs for the covariance between the number of claims. The one on the left-hand side shows $\Cov{N_{i,t}, N_{i,t+j}}$ for $t=1, \ldots, 10$ for lag $j=1, \ldots, 15$, while the one on the right-hand side shows $\Cov{N_{i,t-j}, N_{i,t}}$ for $t = 15$ and for lag $j=1, \ldots, 14$.  
\begin{figure}
	\centering
	\includegraphics[width=8cm]{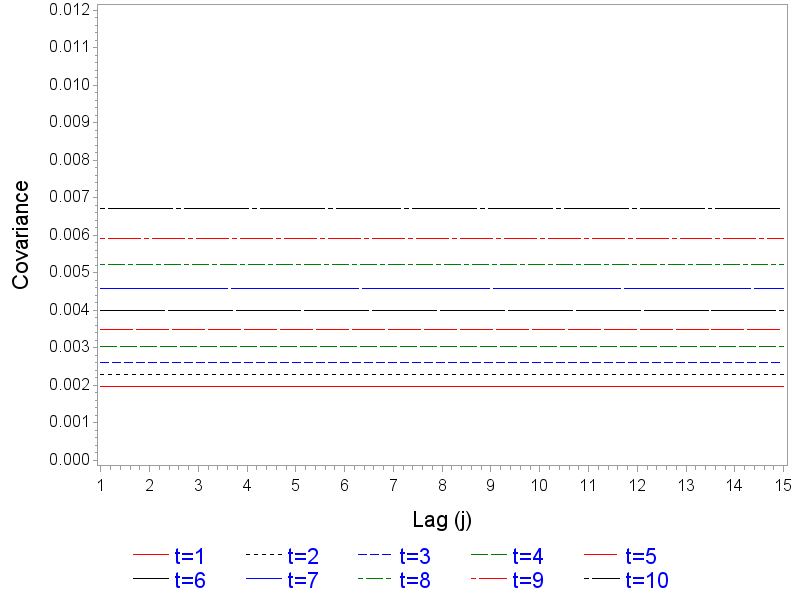}
	\includegraphics[width=8cm]{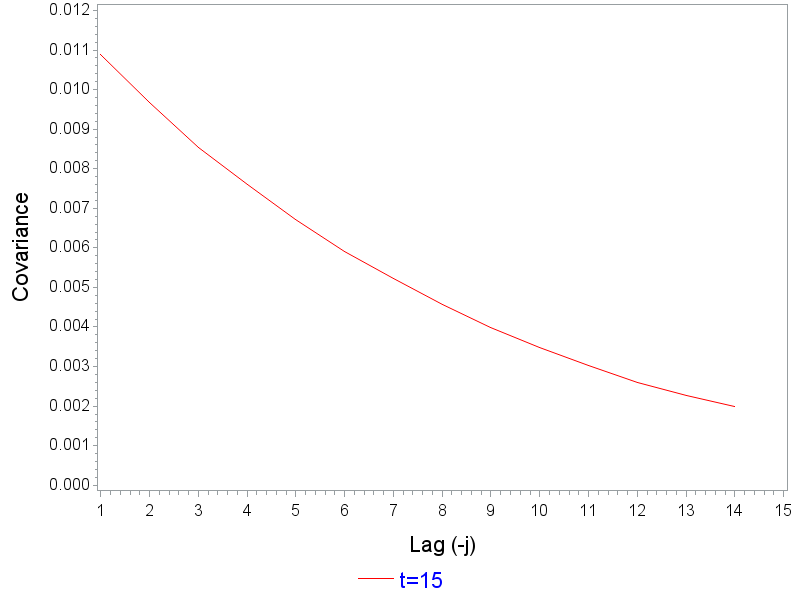}
	\caption{Covariances of the HF models with $\Cov{N_{i,t}, N_{i,t+j}}$ on the left-hand side and $\Cov{N_{i,15-j}, N_{i,15}}$ at right-hand side.}
	\label{fig:covHF}
\end{figure}
Analysis of the graphs indicates, for example, that the number of claims observed in the first contract of insured $i$ will have the same impact on all future contracts.  However, the covariance between claim counts will grow as $t$ becomes larger. Therefore, the following relation holds:
\begin{align*}
\Cov{N_{i,t_1}, N_{i,t_1+j}} \le  \Cov{N_{i,t_2}, N_{i,t_2+j}}, \text{ for } t_1 < t_2.
\end{align*}
The HF model was constructed to give unequal weight on the predictive premium depending on the age of the claim.  Now, we see that the model does not give less weight to older claims, but gives more and more weight to newer claims.  This pattern can be seen more clearly by the following proposition.

\begin{Proposition}\label{propHF}
Based on an insured who has never claimed, the premium increase resulting from a claim will be higher if the insured has a long driving experience. In addition, the greater the driving experience of an insured, the greater the impact of a claim on the following year's premium.
\end{Proposition}

\begin{proof}
First, based on the HF-NBBeta model, let us compute the premium at time $t+1$ for an insured who did not report at all: 

\begin{align*}
      \pi_{i, t+1}
       &= \lambda_{i, t + 1}\left(\frac{\nu^{t} b }{\left(\nu\right)^{t} a + \sum_{k=1}^{t}\left(\nu\right)^{k}\lambda_{i, t-k+1} - 1}\right).
\end{align*}

Then, we compute the premium of an insured who had only one claim at time $w < t$:

	\begin{align*}
      \pi_{i, t+1}^{(w)} 
       &= \lambda_{i, t + 1}\left(\frac{\nu^{t} b  + \nu^{t-w+1}}{\left(\nu\right)^{t} a + \sum_{k=1}^{t}\left(\nu\right)^{k}\lambda_{i, t-k+1} - 1}\right).
        \end{align*}
        
The impact of the claim at time $w$, or the increase in premium, is then calculated as:

\begin{align*}
\frac{\pi_{i, t+1}^{(w)}}{\pi_{i, t+1}}
&= \left(\frac{\nu^{t} b + \nu^{t-w+1}} {\nu^{t} b}\right) = 1 +  \frac{\nu^{-w+1}}{b} = 1 +  \left(\frac{\nu}{b}\right) \nu^{-w} 
\end{align*}

Because $\nu < 1$ for the HF model, which means that 
the premium increase will amplify as $w$ grows.  Moreover, we note that the increase does not depend on $t$ or on the age of the claim 
($t-w$).  In other words, the HF models suppose that the impact of a claim at time $w$ will stay the same for all future contracts $t$.
\end{proof}

In classic panel data models, such as the MVNB or the NBBeta models, it is impossible for an insured who claims at least one time to ever have a premium equal to an insured who has never claimed, even after several years since the first and only claim. With the HF generalization and the introduction of the weight parameter $\nu$, we would expect that the impact of old claims gradually becomes insignificant in future premiums. We just saw that this is not the case.  For the HF-NBBeta used with our data, we obtained $\hat{b}=5.5$ and $\hat{\nu} = 0.9$, which means that a claim at time $w=1$ will cause an  $18.2~\%$ increase in the premium for all future contracts, while a claim at time $w=15$ will always increase the premium by $78.9~\%$ compared with an insured who never files an accident. Even if the fit of the HF models is interesting, this property of the model means that it cannot be seriously considered in practice: 
the longer the insureds' driving experience, the higher their penalty for filing an accident. The BMS-panel model does not have this property: for example, with the maximum number of levels $s=11$, an insured without any claims for $11$~consecutive years will have the same premium as a similar insured who did not report at all.  We think that this property is more realistic, and more desirable for insurers.\\

The BMS-NB1 model with $s=11, \ell^*=1$ and $\Psi=6$ supposes a linear relativity for claim penalty with $\widehat{\delta}= 0.12$.  Depending on the BMS level of each insured, it means that the premium will be equal to $1.12, 1.24, \ldots, 2.08, 2.20$ times the basic premium (level~$1$). For example, for an insured at level~$2$, a claim during the year results in a jump of $\Psi = 6$ levels, which represents an increase of approximately $64\%$ of the premium. Conversely, a year without claim generates a premium reduction of $11\%$.

To evaluate more precisely the impact of the current level and the dependence between claim counts on different contracts, based on the result of Proposition~\ref{prop:cov}, we compute the covariance of the BMS-panel model, as shown in Figure~\ref{fig:cov}. To compare the covariances, we also included the covariance implied by the NBBeta model, which stays constant over time. 
Unlike in the HF model, the covariance depends on $\ell_t$, the BMS level at time $t$. The dynamic property of the BMS-panel model can be observed: the dependence between annual contracts decreases as the lag grows.  We think that this is one of the most important properties of a merit rating plan.  The impact of $\ell_1$ (the level of the BMS at time 1) on the covariance is clear. For example, as an insured located on $\ell_1 = 1$ cannot attain a lower level at time 2, the covariance between $N_1$ and $N_{2}$ is limited. As $\ell_1$ grows, the covariance also increases, but because the impact of one claim on the next BMS level is $\Psi = 6$, we observe that the covariance begins to decrease for $\ell_1 > 5$.  Finally, insured on $\ell_1 = s = 11$ are also limited, this time by the fact that they cannot attain a higher BMS level.  As the lag between two number of claims grows, the covariance decreases and the impact of $\ell_1$, the level of the BMS at time 1, vanishes.

\begin{figure}
	\centering
	\includegraphics[width=8cm]{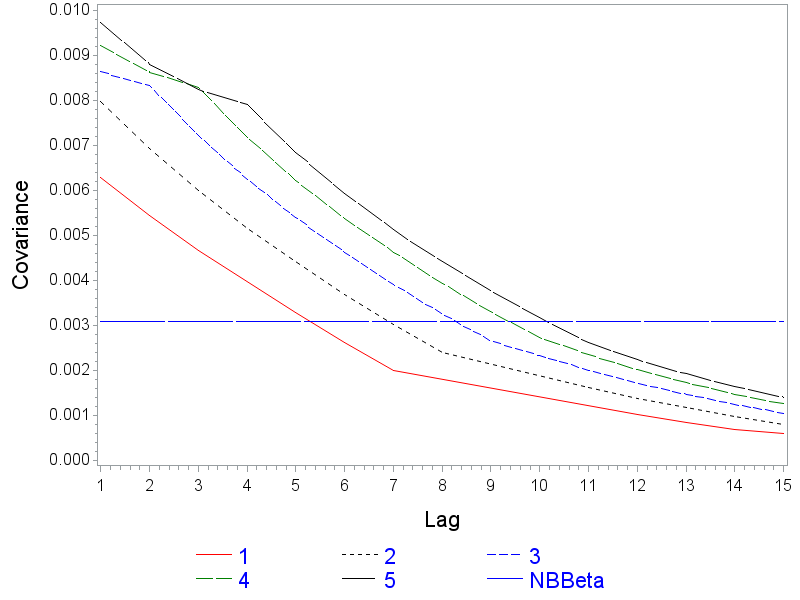}
	\includegraphics[width=8cm]{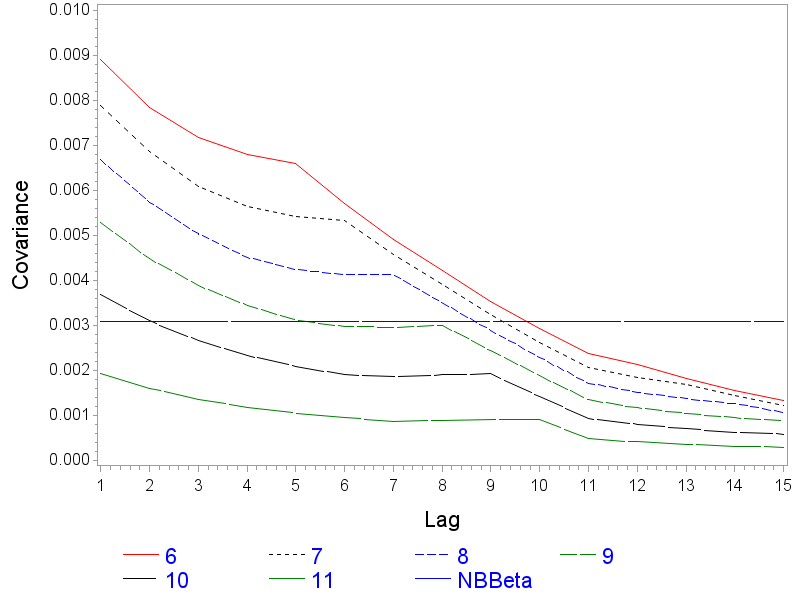}
	\caption{Covariance of the BMS-panel depending on the level $\ell$}
	\label{fig:cov}
\end{figure}

\clearpage
  \section{Conclusion}
 \label{sec:conclusion}

The BMS approach proposed in this paper generates a claim score that is easy to explain. The claim score is also easy to use because it summarizes the whole claim experience of an insured, by also taking into account the age of each claim. While we initially wanted to show that the BMS-panel model was practical and flexible, we were surprised to see that it also generated a much better statistical fit than some of the most popular counting distributions known in the actuarial literature. The decreasing covariance of the model,  which assigns a lower weight to older claims, seems to best explain the interesting predictive power of the model. Indeed, the vast majority of existing models do not show such flexibility. The HF model proposed in this paper, which seemed to be a straightforward generalization of classic panel data models, seem an interesting solution. However, by analyzing the predictive premium of those HF models, we have shown that they become completely unpractical when a claimant's insurance history grows. Indeed, the impact of a single claim for an insured with a long driving experience becomes so important that no insurer would be interested in implementing such an approach.  

Many studies and generalizations of the BMS-panel models are now possible:
\begin{itemize}
\item we can study a two-dimensional generalization of the BMS approach, where several types of claims (such 
  as at-fault or no-fault accidents) could be modeled;
\item achieving a better understanding of the multivariate dynamics in insurance;
\item many properties of the BMS, already developed and understood, can now be applied to BMS-panel model, e.g., asymptotic properties of the BMS, tools to compare merit rating plans, hunger for bonus, etc.
\end{itemize} 

 \section*{Acknowledgements}

Jean-Philippe Boucher and  Mathieu Pigeon would like to thank the financial support from the Canadian Institute of Actuaries for its financial support in the form of research grant \# CS000168.


\begin{thebibliography}{5}

 \bibitem{AA2016}
{\sc A. Abdallah} and {\sc J.-P. Boucher} and {\sc H. Cossette} and {\sc J. Trufin}  (2016). 
\newblock Sarmanov Family of Bivariate Distributions for Multivariate Loss Reserving Analysis.
\newblock {\em North American Actuarial Journal}, 20(2), 184--200.

\bibitem{Alb}
{\sc P. Albrecht} (1985). 
\newblock An evolutionary credibility model for claim numbers
\newblock {\em  ASTIN Bulletin}, 15(1), 1--17.

\bibitem{B2007}
{\sc C. Bolancé} and {\sc M. Denuit} and {\sc M. Guillén} and {\sc P. Lambert}  (2007). 
\newblock Greatest accuracy credibility with dynamic heterogeneity: the Harvey-Fernandes model.
\newblock {\em Belgian Actuarial Bulletin}, 7(1), 14--18.

\bibitem{Boucher2006}
{\sc Boucher, J.-P.} and {\sc M. Denuit} (2006). 
\newblock Fixed versus Random Effects in Poisson Regression Models for Claim Counts: Case Study with Motor Insurance.
\newblock {\em ASTIN Bulletin}, 36, 285--301.

\bibitem{Boucher2008}
{\sc Boucher, J.-P.} and {\sc M. Denuit} and {\sc M. Guillén} (2008). 
\newblock  Models of Insurance Claim Counts with Time Dependence Based on Generalisation of Poisson and Negative Binomial Distributions.
\newblock {\em Variance}, 2(1), 135--162.

\bibitem{Boucher2009}
{\sc Boucher, J.-P.} and {\sc M. Guillén} (2009). 
\newblock A Survey on Models for Panel Count Data with Applications to Insurance. 
\newblock {\em Revista de la Real Academia de Ciencias Exactas, Físicas y Naturales}, 103(2), 277--295.

\bibitem{Boucher2014}
{\sc Boucher, J.-P.} and {\sc Inoussa, R.} (2006). 
\newblock A Posteriori Ratemaking with Panel Data.
\newblock {\em ASTIN Bulletin}, 44(3), 587--612.

\bibitem{BG}
{\sc Bühlmann, H.} and {\sc Gisler, A.} (2005).
\newblock {\em A Course in Credibility Theory and its Applications\/}.
\newblock Springer Berlin Heidelberg New York

\bibitem{CHO}
{\sc Choirat, C.} and {\sc Raffaello, S.} (2012).
\newblock {\em Estimation in Discrete Parameter Models}.
\newblock {\em Statistical Science}, 27(2), 278--293.

\bibitem{Denuit}
{\sc Denuit, M.} and {\sc Maréchal, X.} and {\sc Pitrebois, S.} and {\sc Walhin, J.-F.} (2007). 
\newblock {\em Actuarial Modelling of Claim Counts: Risk Classification, Credibility and Bonus-Malus Scales}. 
\newblock Wiley: New York.

\bibitem{frees2014}
 {\sc Frees, E.W. and Derrig, R.A. and Meyers, G.} (2014).
\newblock {\em Predictive modeling applications in actuarial science},
\newblock Cambridge University Press.

\bibitem{Gilde}
{\sc Gilde, V.} and {\sc Sundt, B.} (1989). 
\newblock On Bonus systems with credibility scales. 
\newblock {\em Scandinavian Actuarial Journal}, 1989(1), 13--22.

\bibitem{GourierouxJasiak2004}
{\sc Gourieroux, C.} and {\sc Jasiak, J.} (2004).
\newblock Heterogeneous INAR(1) Model with Application to Car Insurance.
\newblock {\em Insurance: Mathematics and Economics}, 34(2), 177--192.

\bibitem{H1989}
{\sc A.C. Harvey} and {\sc C. Fernandes} (1989).
\newblock Time series models for count or qualitative observations.
\newblock {\em Journal of Business \& Economics Statistics}, 7, 407--422.

\bibitem{HAM}
  {\sc Hammersley, J. M.} (1950).
  \newblock On estimating restricted parameters (with discussion).
  \newblock {\em J. Roy. Statist. Soc. Ser. B}, 12, 192--240.

\bibitem{Jung}
{\sc R. Jung} and {\sc R. Liesenfeld} (2001).
\newblock stimating time series models for count data using efficient importance sampling.
\newblock {\em AStA Advances in Statistical Analysis}, 4(85), 387--407.

\bibitem{Lemaire}
{\sc Lemaire, J.} (1995). 
\newblock {\em Bonus-Malus Systems in Automobile Insurance}. 
\newblock Kluwer Academic Publisher, Boston.

\bibitem{McCullaghNelder1989}
{\sc McCullagh, P.} and {\sc Nelder, J.~A.} (1989).
\newblock {\em Generalized Linear Models\/}.
\newblock London : Chapman and Hall, 2nd ed.

\bibitem{Pinquet}
{\sc Pinquet, J.} and {\sc Guillén, M.} and {\sc Bolancé, C.} (2001).
\newblock Allowance for the age of claims in bonus-malus systems.
\newblock {\em  ASTIN Bulletin}, 31(2), 337-348.

\bibitem{Shi2014}
{\sc Shi, P.} and {\sc Valdez, E.} (2016).
\newblock  Longitudinal modeling of insurance claim counts using jitters.
\newblock {\em Scandinavian Actuarial Journal}, 2014(2), 159-179.

\bibitem{Shi2016}
{\sc Shi, P.} and {\sc Feng, X.} and {\sc Boucher, J.-P.} (2016).
\newblock Multilevel modeling of insurance claims using copulas.
\newblock {\em  The Annals of Applied Statistics}, 10(2), 834-863.

\bibitem{Wink}
{\sc Winkelmann, R.}  (2010).
\newblock {\em Econometric Analysis of Count Data\/}.
\newblock Springer-Verlag Berlin Heidelberg, 5th ed.


\end{thebibliography}
\end{document}